\documentclass[preprint,12pt]{elsarticle}

\usepackage{amsmath,amssymb,amsthm}

\newtheorem{theorem}{Theorem}
\newtheorem{lemma}[theorem]{Lemma}

\newtheorem{definition}[theorem]{Definition}

\newcommand{\Q}{\mathbb{Q}}

\newcommand{\Z}{\mathbb{Z}}

\journal{Information Processing Letters}

\begin{document}

\begin{frontmatter}

\title{NP-hardness of $p$-adic linear regression}

\author{Greg Baker}
\ead{greg.baker@anu.edu.au}
\affiliation{organization={Australian National University},
            city={Canberra},
            country={Australia}}

          \begin{abstract}
            $p$-adic linear regression is the problem of finding coefficients
            $\beta$ that minimise
  $\sum_i |y_i - x_i^\top\beta|_p$. We prove that computing an optimal
  solution is NP-hard via a
  polynomial-time reduction from Max Cut using a regularisation
  gadget. 
\end{abstract}

\begin{keyword}
$p$-adic analysis \sep linear regression \sep computational complexity \sep NP-hardness \sep Max Cut
\end{keyword}

\end{frontmatter}

\section{Introduction}

The $p$-adic numbers, introduced by Hensel in 1897, provide an alternative
completion of the rationals with respect to a non-Archimedean absolute
value. While their applications in number theory and algebraic geometry
are classical, recent work has explored their use in machine learning,
where the ultrametric structure of $p$-adic spaces offers natural
representations for hierarchical and tree-structured data
\cite{murtagh2004,bradley2008,khrennikov2013,martins2025padic,zubarev2025padic,aaclpadiclinear}.

A fundamental task in this setting is \emph{$p$-adic linear regression}:
given data points $(x_i, y_i)$ with $x_i \in \Q^n$ or $\Z^n$ and $y_i \in \Q$ or $\Z$,
find coefficients $\beta \in \Q^n$ minimising
\[
L(\beta) = \sum_{i=1}^{m} |y_i - x_i^\top \beta|_p,
\]
where $|\cdot|_p$ denotes the $p$-adic absolute value.

To our knowledge, there are no known examples of $p$-adic linear
regression arising as an oracle (black-box) subroutine problem.  The
real-world applications of $p$-adic methods we are aware of have
instead been posed directly as optimisation problems---for example,
grammar morphology \cite{aaclpadiclinear} and ontology-related
optimisation \cite{baker2025padic}.  Accordingly, we formulate our
complexity result for the optimisation version of $p$-adic linear
regression.

This also highlights a striking contrast with classical regression: in
Euclidean settings, least-squares regression admits a closed-form
solution and $\ell_1$ regression can be solved in polynomial time via
convex optimisation.  Here, changing only the loss to the ultrametric
$|\cdot|_p$ yields an NP-hard optimisation problem.

This problem differs fundamentally from classical least-squares regression.
The ultrametric inequality $|a + b|_p \leq \max(|a|_p, |b|_p)$ creates
discrete, hierarchical structure in the loss landscape. Prior work
\cite{baker2025padic} established that optimal solutions pass through at least
$n+1$ data points under non-degeneracy conditions, yielding polynomial-time
algorithms for fixed dimension.

In this paper, we complete the complexity picture by
proving that $p$-adic linear regression is NP-hard when the
dimension is part of the input (Section~\ref{sec:hardness}).

\section{Preliminaries}

\subsection{The $p$-adic absolute value}\label{padicdef}

For a prime $p$ and a non-zero rational $x = p^k \cdot (a/b)$ where
$\gcd(ab, p) = 1$, the \emph{$p$-adic valuation} is $v_p(x) = k$, and
the \emph{$p$-adic absolute value} is $|x|_p = p^{-v_p(x)}$. By
convention, $|0|_p = 0$.

The $p$-adic absolute value satisfies the \emph{ultrametric inequality}:
\[
|x + y|_p \leq \max(|x|_p, |y|_p),
\]
with equality when $|x|_p \neq |y|_p$. This non-Archimedean property
distinguishes $p$-adic analysis from real analysis and is central to
our results.

\subsection{Problem definitions}

\begin{definition}[$p$-adic linear regression]
\leavevmode\\
Given data points $(x_1, y_1), \ldots, (x_m, y_m)$ with $x_i \in \Z^n$
(or $\Q^n$) and $y_i \in \Z$ (or $\Q$), and a prime $p$, find $\beta \in
\Q^n$ minimising
\[
\sum_{i=1}^{m} |y_i - x_i^\top \beta|_p.
\]
\end{definition}

\begin{definition}[Max Cut]
Given an unweighted graph $G = (V, E)$, find a partition of $V$ into
sets $S$ and $T$ that maximises the number of edges with one endpoint
in $S$ and the other in $T$.
\end{definition}

Max Cut is NP-hard \cite{karp1972,garey1979}. It is convenient to talk
of ``colouring'' the vertices one of two colours and to describe edges
as monochromatic or bichromatic. The vertices of one colour are put
into one set and those of the other colour into the other set.

\section{NP-hardness of $p$-adic regression}
\label{sec:hardness}

We reduce Max Cut to $2$-adic linear regression. 

\begin{theorem}\label{thm:main}
$p$-adic linear regression is NP-hard (already for $p=2$).
\end{theorem}

\begin{proof}
Let $G = (V, E)$ be an unweighted graph with $|V| = n$ vertices and
$|E| = m$ edges. We construct a regression instance with $n$
coefficients $\beta_1, \ldots, \beta_n$ (one per vertex).

\textbf{Edge points.} For each edge $(u, v) \in E$, add one point
$(e_u + e_v, 1)$ where $e_u$ is the $u$-th standard basis vector.
These contribute to the residual $\left| 1 - \beta_u - \beta_v \right|_2$.

Note that if $\beta_u = \beta_v = 1$ then the residual $\left| -1
\right|_2 = 1$; likewise if $\beta_u = \beta_v = 0$ then the residual
$\left| 1 \right|_2 = 1$.  Thus, if $\beta_u$ and $\beta_v \in
\{0,1\}$ the residual counts the number of monochromatic edges.

Unfortunately, nothing constrains $\beta_u$ to $\{0,1\}$ yet.
For this we add regularisation.

Set $M := m+1$. We create the following data points:

\textbf{Forcing points.} For each vertex $j \in \{1, \ldots, n\}$, add
$M$ copies of $(e_j, 0)$ and $M$ copies of $(e_j, 1)$.
These contribute residuals $|\beta_j|_2$
and $|\beta_j - 1|_2$ respectively.

This is the $p$-adic equivalent of regularisation in a Euclidean
linear regression problem: including a multiple of the coefficient's values
in the loss.

The total loss for our model is then:
\[
L(\beta) = M \sum_{j=1}^{n} \bigl(|\beta_j|_2 + |\beta_j - 1|_2\bigr)
+ \sum_{(u,v) \in E} |\beta_u + \beta_v - 1|_2.
\]

We will demonstrate that the $n$ optimal coefficients ($\beta$) that
minimise this loss correspond to the colourings for the $n$
vertices in the Max Cut problem. To do this we need to show
that $\beta_j \in \{0,1\}$ for all $j$ (so that they correspond to valid
colourings) and that a minimal loss is also a maximal solution to Max Cut.

We establish four lemmas.

\begin{lemma}\label{lem:toobig}
  Any solution to the $p$-adic linear regression problem specified above
  which leads to a loss equal to or greater than $M (n + 1)$ cannot
  be optimal.
\end{lemma}


\begin{proof}
Let us start with a naive solution: $\beta^{0}$ (where $\beta_j = 0$ for all $j$), then we can calculate

\[
L(\beta^{0}) = M \sum_{j=1}^n \left( \left|0\right|_2 + \left| -1 \right|_2 \right) + \sum_{(u,v) \in E} \left| -1 \right|_2 = M n + m
\]

Therefore, any solution which leads to a loss greater than $M n + m$ cannot
be the optimal solution. Since $m < M$ we can further say that any
solution with a loss equal to or greater than $M (n + 1)$ cannot be the optimal solution.
\end{proof}

\begin{lemma}\label{lem:binary}
For any $a \in \Q$, we have $|a|_2 + |a-1|_2 \geq 1$, with equality
if and only if $a \in \{0, 1\}$.
\end{lemma}

\begin{proof}
Since $|1|_2 = |a - (a-1)|_2 \leq \max(|a|_2, |a-1|_2)$, at least one
of the terms is $\geq 1$. For the sum to equal 1, the other term must
be 0, forcing $a = 0$ or $a = 1$.
\end{proof}

\begin{lemma}
  If for some $\beta, k$ we have
  $\left| \beta_k \right|_2 > 1$ or
  $\left| \beta_k - 1 \right|_2 > 1$, then
  $\beta$ is not an optimal solution.
\end{lemma}

\begin{proof}
  Reviewing the formula for $| \cdot |_2$, it is impossible to get a negative
  measure; there is also no $a$
  that can exist for which $1 < | a |_2 < 2$. Thus
  $\left| \beta_k \right|_2 \geq 2$.

  Substituting these inequalities and using Lemma~\ref{lem:binary}, we
  have
  \[
  \begin{split}
  L(\beta)
  &\ge M \sum_{j\ne k} \bigl(|\beta_j|_2 + |\beta_j - 1|_2\bigr)
      + M\bigl(|\beta_k|_2 + |\beta_k - 1|_2\bigr) \\
  &\ge M(n-1) + 2M \;=\; M(n+1),
  \end{split}
  \]
  where we also use that the edge-term sum is nonnegative.

  Thus by Lemma \ref{lem:toobig}, this would not be an optimal solution.
\end{proof}

Since $\left| \beta_j \right|_2 \leq 1$ for all $j$ in an optimal solution
we know that the denominators of $\beta_j$ in reduced form are not divisible
by 2.

Thus, by Lemma 4.3.2 in Gouv\^{e}a~\cite{Gouvea2020padic}, we can represent
any $\beta_j$ as a sum of increasing powers of 2, where each $a_t \in \{0,1\}$.

\[
\beta_j = a_0 + 2 a_1 + 4 a_2 + \ldots
\]

Note that 2-adically, the later terms are getting steadily smaller. If
we need to use an infinite number of terms to represent $\beta_j$ (as
we would if it is a non-integer), the limits are well-defined and we
can approximate it by cutting off after a finite number of terms.

We now define $\bmod 2$ as taking the first term ($a_0$) in that
series. This is exactly the ordinary meaning of $\bmod 2$ for
integers, but is also well-defined for non-integers.

\begin{lemma}\label{lem:rounding}
For any $\beta \in \Z_2^n$, let $s_j = \beta_j \bmod 2 \in \{0,1\}$.
Then $L(s) \leq L(\beta)$.
\end{lemma}

\begin{proof}
The forcing terms satisfy $|s_j|_2 + |s_j - 1|_2 = 1 \leq |\beta_j|_2 + |\beta_j - 1|_2$
by Lemma~\ref{lem:binary}.

For edge $(u,v)$: if $s_u = s_v$, then $\beta_u + \beta_v - 1$ is odd,
so $|\beta_u + \beta_v - 1|_2 = 1 = |s_u + s_v - 1|_2$. If $s_u \neq s_v$,
then $\beta_u + \beta_v - 1 \equiv 0 \pmod{2}$, so
$|\beta_u + \beta_v - 1|_2 \leq 1/2$, while $|s_u + s_v - 1|_2 = |0|_2 = 0$.
\end{proof}

By Lemma~\ref{lem:rounding}, there exists an optimal solution in
$\{0,1\}^n$. For $s \in \{0,1\}^n$, an edge $(u,v)$ contributes:
\begin{itemize}
\item 0 to the loss if $s_u \neq s_v$ (edge crosses the cut), since
$s_u + s_v = 1$.
\item 1 to the loss if $s_u = s_v$ (edge does not cross), since
$s_u + s_v \in \{0, 2\}$ and $|{\pm}1|_2 = 1$.
\end{itemize}

Thus
\[
L(s) = M n + (\text{number of non-crossing edges}) = M n + m - \text{cutsize}(s).
\]

Since the term $Mn+m$ does not depend on $s$, minimising $L(s)$ is
equivalent to maximising $\text{cutsize}(s)$. In particular, from an
optimal regression value $L^\star$ we can recover the maximum cut size
as $\text{cutsize}^\star = Mn + m - L^\star$. Equivalently, for any
integer $K$, $G$ has a cut of size at least $K$ if and only if the
constructed regression instance admits $\beta$ with $L(\beta)\le
Mn+m-K$. Thus any Max Cut problem can be turned into a $p$-adic linear
regression problem, and if $p$-adic linear regression can be solved in
polynomial time, then $P=NP$.
\end{proof}

NP-hardness is a worst-case classification: it does not preclude fast
algorithms on structured instances, and it motivates approximation
algorithms and the identification of tractable regimes.

\section{Discussion: tractable regimes and open problems}

Theorem~\ref{thm:main} explains the worst-case complexity when the
dimension is part of the input. It complements known polynomial-time
methods for fixed dimension \cite{baker2025padic}, and suggests several
natural parameters that affect computational behaviour.

\subsection{Fixed dimension}\label{fixed-dimension}

When the dimension $n$ is fixed, Baker \emph{et al.}\ show that (under a
non-degeneracy condition) an optimal solution passes through at least
$n+1$ data points \cite{baker2025padic}. Enumerating all such candidate
hyperplanes yields an $O(m^{n+1})$ algorithm.

\subsection{Dependence on $p$ and the dataset}

The dependence on $p$ enters through divisibility of residuals. For an
integer residual $r$, we have $|r|_p=0$ if $r=0$ and $|r|_p=1$ whenever
$p\nmid r$. Thus, if the data are integral and a candidate fit produces
integer residuals whose (ordinary) magnitudes are all $<p$, then every
nonzero residual contributes $1$ and the loss simply counts the number
of nonzero residuals. In this regime, minimising $p$-adic loss is
equivalent to \emph{exact fitting}: maximising the number of data
points interpolated exactly. In two dimensions this is the maximum
collinear points problem; near-quadratic algorithms are known
\cite{chen2022fitting} where the algorithm of \cite{baker2025padic} would be
cubic in the number of data points.

\subsection{Alternative aggregations}

Martins~\cite{martins2025padic} studies $p$-adic learning with a $\max$
aggregation rather than a sum. The complexity of $p$-adic regression
under this and other aggregations is currently open.

\subsection{What happens with decreasing regularisation?}

In Section \ref{sec:hardness} we needed the regularising terms in order
to bind $\beta \in \{0,1\}^n$. The reduction uses $M=m+1$ copies of each
forcing point, but without any forcing points the constructed instance
is trivial: setting $\beta_j=\frac{1}{2}$ for all $j$ makes every edge
residual zero. Understanding how the computational complexity changes
as the strength of this regularisation decreases is an open problem.

\section{Conclusion}

We have shown that $p$-adic linear regression is NP-hard when the
dimension is part of the input, completing its complexity
classification alongside the known polynomial-time algorithm for fixed
dimension. The reduction from Max Cut leverages the ultrametric
inequality in a fundamental way: the regularisation terms constrain
$\beta$ to binary values precisely because of the non-Archimedean
property.

Open questions include approximation algorithms for $p$-adic linear
regression, how complexity varies with the choice of $p$ and the
regularisation strength, and the complexity under alternative loss
aggregations (max vs.\ sum).

\end{document}